\documentclass[10pt,conference,letterpaper]{IEEEtran}
\usepackage{cite}
\usepackage{amsmath,amssymb,amsfonts,amsthm}
\usepackage{algorithmic}
\usepackage{graphicx}
\usepackage{textcomp}
\usepackage{xcolor}
\usepackage{float}
\usepackage{pgfplots}
\usepackage{tikz}
\usepackage{siunitx}
\usepackage[nolist]{acronym}
\usepackage[font=footnotesize]{caption}

\usepackage[normalem]{ulem}

\newcommand{\x}{\mathbf{x}}

\newcommand{\kw}{\mathbf{k}}
\newcommand{\lspace}{\mathcal{L}^2}

\newcommand{\bspace}{\mathcal{B}_\mathcal{S} }

\newcommand{\PS}{P_{\mathcal{S}}}

\newcommand{\lambdam}{\lambda_n^{(m)}}

\newcommand{\Dm}{\chi_{\mathcal{D}_m}}

\newcommand{\Thm}{T_h^{(m)}}

\newcommand{\Thmh}{\hat{T}_h^{(m)}}
\newcommand{\Thh}{\hat{T}_h}
\newcommand{\fm}{f}
\newcommand{\gm}{g}
\newcommand{\anm}{a_n^{(m)}}
\newcommand{\bnm}{b_n^{(m)}}
\newcommand{\cnm}{c_n^{(m)}}

\newcommand{\pro}{\psi_n^{(m)}}

\newtheorem{theorem}{Theorem}
\newtheorem{proposition}{Proposition}

\def\BibTeX{{\rm B\kern-.05em{\sc i\kern-.025em b}\kern-.08em
    T\kern-.1667em\lower.7ex\hbox{E}\kern-.125emX}}
\begin{document}

\title{Extrapolation of Bandlimited Multidimensional Signals from Continuous Measurements\\
\thanks{Identify applicable funding agency here. If none, delete this.}
}

\author{
	\IEEEauthorblockN{Cornelius Frankenbach\IEEEauthorrefmark{1}\IEEEauthorrefmark{2}, Pablo Mart\'{i}nez-Nuevo\IEEEauthorrefmark{1}, Martin M\o ller\IEEEauthorrefmark{1}, Walter Kellermann\IEEEauthorrefmark{2}}
	
	\IEEEauthorblockA{\IEEEauthorrefmark{2}Multimedia Communications and Signal Processing, University of Erlangen-Nuremberg}
	\IEEEauthorblockA{\IEEEauthorrefmark{1}Bang \& Olufsen a/s, R\&D Acoustics}
	}

\maketitle

\begin{acronym}
	\acro{pli}[PLI]{Projected Landweber Iteration}
	\acro{dpli}[DPLI]{Damped Projected Landweber Iteration}
	\acro{pswf}[PSWF]{Prolate Spheroidal Wave Function}
	\acro{dft}[DFT]{Discrete Fourier Transform}
	\acro{idft}[IDFT]{Inverse Discrete Fourier Transform}
	\acro{nmse}[NMSE]{Normalized Mean Square Error}
	\acro{snr}[SNR]{Signal-to-Noise Ratio}
\end{acronym}

\begin{abstract}
	Conventional sampling and interpolation commonly rely on discrete measurements. In this paper, we develop a theoretical framework for extrapolation of signals in higher dimensions from knowledge of the continuous waveform on bounded high-dimensional regions. In particular, we propose an iterative method to reconstruct bandlimited multidimensional signals based on truncated versions of the original signal to bounded regions---herein referred to as continuous measurements. In the proposed method, the reconstruction is performed by iterating on a convex combination of region-limiting and bandlimiting operations. We show that this iteration consists of a firmly nonexpansive operator and prove strong convergence for multidimensional bandlimited signals. In order to improve numerical stability, we introduce a regularized iteration and show its connection to Tikhonov regularization. The method is illustrated numerically for two-dimensional signals.
\end{abstract}


\section{Introduction}
Sampling and reconstruction of one-dimensional bandlimited signals based on discrete samples has been widely studied in signal processing. This extends from the conventional sampling theorem \cite{Shannon:1949aa}, where discrete samples are taken at equally-spaced time instants, to more general irregular sampling schemes \cite{Yen:1956aa,Paley:1934aa}. Further generalizations to other function spaces can be found, for example, in connection to wavelets or parametrized functions \cite{Eldar:2015aa}. 

In contrast, we propose a sampling and reconstruction paradigm for bandlimited multidimensional signals where measurements consist of truncated versions of the original signal, i.e., they are based on knowledge of the full waveform on a predefined set of bounded regions. We will refer to these as continuous measurements. In particular, consider the function $\chi_{\mathcal{D}}$ which evaluates to 1 if its argument is in the set $\mathcal{D}$, and 0 otherwise. Then, the measurements of the source signal $h$ are given by the set $\{\chi_{\mathcal{D}_m}(\x)\cdot h(\x)\}_{m=1}^M$ where $\mathcal{D}_m$ corresponds to a bounded region, $M$ is a positive integer denoting the total number of measurements, and $\mathbf{x}\in\mathbb{R}^N$ for some $N\geq1$. Fig.~\ref{fig:Truncation} shows an example of the sampling regions considered in this framework. In practice, conducting or processing a continuous measurement is non-trivial and is beyond the scope of this paper. We work on the assumption that the underlying signal $h$ is continuously known inside all regions $\mathcal{D}_m$.

\begin{figure}[!t]
	\centering
	\includegraphics[width = 0.3\textwidth]{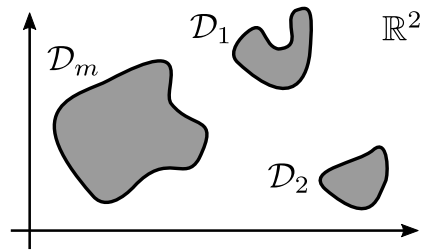}
	\setlength{\belowcaptionskip}{-20pt}
	\caption{Continuous measurements on $M\geq1$ non overlapping regions $\mathcal{D}_m$.}
	\label{fig:Truncation}
\end{figure}

In some settings, these continuous measurements can be further alleviated by only considering the values along the contours of the regions. For example, in acoustics, optics, or seismology \cite{Williams:1999aa, Born:2013aa,schleicher2007seismic}, the Kirchhoff-Helmholtz integral equation fully determines the values inside a region based on the values on its contour. Additionally, continuous measurements can be viewed as an extreme case of discrete sampling by highly oversampling confined regions in space. This can be of particular interest in room acoustics where it is becoming common that commercial loudspeakers are equipped with an increasing number of built-in microphones. Many of these multidimensional functions are approximately bandlimited and their reconstruction is useful in many applications \cite{Grande:2019aa,Chai:2000aa}.

At first sight, recovery of a function in the entire space, from continuous measurements in a subspace, seems hopeless since some information about the signal is expected to be lost and the solution to the reconstruction problem may not be unique unless further knowledge of the signal structure is available. However, the main observation underlying our approach is that bandlimited signals---in any dimension---can be viewed as analytic functions on their whole domain \cite{Paley:1934aa}. By analytic continuation \cite{Stein:2003aa}, a unique representation by their values, for example, on a bounded region, is then guaranteed.

Within the context of spectral estimation, Papoulis \cite{papoulis1975new} exploited this property by proposing an iterative method to reconstruct a one-dimensional bandlimited signal from a single segment. In \cite{cadzow1979extrapolation}, this approach was later formalized within operator theory where a one-step procedure of Papoulis' iteration was derived for a restricted class of bandlimited signals. The first connection to nonexpansive operators was introduced in \cite{tom1981convergence} but still in the one-dimensional case. In connection to the uncertainty principle, it was shown in \cite{donoho1989uncertainty} that stable reconstruction using Papoulis' algorithm is possible, for example, whenever the missing information of the signal of interest lies on finite segments. Note that the latter assumes knowledge of the whole signal except for some missing segments as opposed to known segments as in \cite{papoulis1975new}. Papoulis' algorithm can also be viewed as an alternating projection method that has been applied to more general scenarios assuming different a priori knowledge of the signal \cite{landau1961recovery, youla1982image, sezan1982image, schafer1981constrained}. 

In this paper, we prove reconstruction of bandlimited multidimensional signals from multiple continuous measurements, i.e., from individual---possibly weighted---truncated versions of the original signal. We first formulate the problem in Section \ref{sec:prob}. In Section \ref{sec:def}, we propose an iterative method, consisting of a firmly nonexpansive operator that inverts these truncation operators. We prove strong convergence of this sequence of iterates. In Section \ref{sec:Regularization}, we introduce a regularized version of this method in order to improve robustness. We illustrate our theoretical results numerically in Section \ref{sec:simulations} by providing some examples in the two-dimensional case.

% SECTION 2
\section{Problem formulation}
\label{sec:prob}
We aim at reconstructing a bandlimited signal $h:\mathbb{R}^N\to\mathbb{R}$ from its continuous measurements. In other words, the goal is to extrapolate a bandlimited function $h$ to its entire domain $\mathbb{R}^N$ when $h$ is only known on a bounded set $\mathcal{D}\subset\mathbb{R}^N$, where $\mathcal{D}=\bigcup_{m=1}^{M}\mathcal{D}_m$.

We consider a function bandlimited if its Fourier transform has compact support, and we denote the closed convex set of bandlimited functions as
\begin{equation}
	\label{eq:bspace}
	\bspace= \left\{f \in \lspace(\mathbb{R}^N):f(\x)=(2\pi)^{-N}\int_{\mathcal{S}}^{}F(\kw)e^{j\x\cdot\kw}d\kw\right\}
\end{equation}
where $F(\kw)$ is the $N$-dimensional Fourier transform of $f$ with frequency variable $\kw$ and $\mathcal{S}\subset\mathbb{R}^N$ is a compact set. Thus, this definition naturally includes, for example, the notion of bandpass signals.

Note that, by analytic continuation, the truncation operation $\chi_\mathcal{D}$ provides a unique representation of a bandlimited signal. In the following Section, we propose a sequence of iterates $f_k$ that invert this truncation operation and converge strongly to the original signal, i.e., $f_k\to h$.

% SECTION 3
\section{Extrapolation algorithm}
\label{sec:def}
In order to perform reconstruction, based on the measurements $\{\chi_{\mathcal{D}_m}(\x)\cdot h(\x)\}_{m=1}^M$, we define the operator $T_h:\ \mathcal{B}_\mathcal{S}\rightarrow \mathcal{B}_\mathcal{S}$ as follows:
\begin{equation}
	\label{eq:WeightedOperator}
		f\mapsto T_hf:=\sum_{m=1}^{M}\omega_mP_\mathcal{S}(f+\Dm(h-f))
\end{equation}
where $\PS$ is a projection onto the set $\bspace$ and $\sum_{m=1}^M\omega_m=1$ for $\omega_m\in(0,1]$. Then, we introduce the following iteration for $k\geq0$ 
\begin{equation}
	\label{eq:Iteration}
		f_{k+1}:=T_hf_k = \sum_{m=1}^{M}\omega_mP_\mathcal{S}(f_k+\Dm(h-f_k))
\end{equation}

where the iteration starts by using the bandlimited truncated equivalent of $h$, consisting of the continuous measurements of the function $h$ on the corresponding bounded regions, i.e., $f_0=\sum_{m=1}^{M}\PS(\Dm(h))$.

The idea is to iteratively update the signal in the entire domain using the observed signal $h$ in the regions and the bandlimiting operation. Every iteration consists of two steps: First, the current estimate $f_k$ is updated by adding the term $\omega_m(h-f_k)$ within the respective regions $\mathcal{D}_m$. Then, the updated estimate is bandlimited, e.g., by lowpass filtering. Assuming a one-dimensional signal and a single interval, the above iteration reduces to the one presented in \cite{papoulis1975new}.

In the remainder of this Section, we prove strong convergence of the iteration to a unique fixed point, corresponding to the original signal $h$. We show its relation to the \acl{pli}, introduce a regularization term, and discuss the relationship between the continuous measurements and their weights $\omega_m$.

\subsection{Convergence}
\label{sec:proof}
First, let us introduce the operator---with the same domain and co-domain---acting only on a single region, i.e., $\Thm f:=\PS(f+\Dm(h-f))$. The strategy is to first show in Proposition \ref{prop:uniqueness} that an iteration of the form presented above using $T_h^{(m)}$ has a unique fixed point in the space of bandlimited functions. The underlying concept here is again analytic continuation. Then, we show in Theorem \ref{theorem:FirmlyStrongly} that the iteration in (\ref{eq:Iteration}) is actually a convex combination of firmly nonexpansive operators also resulting in a firmly nonexpansive operator. Then, strong convergence to a unique fixed point is guaranteed by the corresponding additional properties.

We will use the well-known fact from operator theory that a sequence of iterates $(T^nx)_{x\in\mathbb{N}}$ converges weakly, denoted by $\rightharpoonup$, to a fixed point if $T$ is firmly nonexpansive and  the set of all fixed points, denoted by $\text{Fix}\ T$, is nonempty \cite{bauschke2012firmly}. The standard inner product in $\lspace$ is denoted by $\langle \cdot,\cdot\rangle$ and the induced norm $\|f\|=\sqrt{\langle f,f\rangle}$.

% PROPOSITION 1
\begin{proposition}
	\label{prop:uniqueness}
	The function $h$ is the unique fixed point for the operator $\Thm$, i.e., $\mathrm{Fix\ }\Thm=\{h\}$.
\end{proposition}
\begin{proof}
	First, we note that
	\begin{equation}
	\label{eq:fixed}
	\Thm h = \PS(h+(h-h)\Dm)=\PS h = h
	\end{equation}
	where the last step follows from the assumption that $h\in\bspace$. Now we need to show that $\mathrm{Fix\ }\Thm$ consists of a singleton by contradiction. Assume there exists a function $f\neq h$ satisfying (\ref{eq:fixed}). This implies that $||h\Dm-f\Dm||=0$. Since there always exists a continuous representative in the corresponding equivalence classes, we have that $||h-f||=0$ by analytic continuation. Thus, $h\equiv f$. 
\end{proof}

The following result shows the convergence of our proposed iteration. Alongside this, we also prove strong convergence based on single bounded regions, i.e., using $T_h^{(m)}$. Interestingly, the latter generalizes and formalizes the convergence results in \cite{papoulis1975new}.

% PROPOSITION 2
\begin{theorem}
	\label{theorem:FirmlyStrongly}
	The sequence of iterates presented in (\ref{eq:Iteration}) converges strongly to the function $h$, i.e., $f_k\to h$.
\end{theorem}
\begin{proof}
	First, we focus on results concerning a single iteration $f_{k+1}^{(m)}:=T_h^{(m)}f_{k}^{(m)}$ where we only use a single bounded region $\mathcal{D}_m$. Let us start by evaluating the expression 
	\begin{equation}
	\label{eq:firmly_left}
	\|\Thm \fm-\Thm \gm\|^2=\sum_{n=0}^{\infty}(\anm-\bnm)^2(1-\lambdam)^2
	\end{equation}
	where we used the fact that $f, g\in\bspace$ can be expanded into \acp{pswf} $\pro$ with eigenvalue $\lambdam$ \cite{slepian_I, slepian_IV, bertero1998introduction} such that $f=\sum_{n=0}^{\infty}\anm\pro$ and $g=\sum_{n=0}^{\infty}\bnm\pro$. We can then write
	\begin{equation}
	\label{eq:firmly_right}
	\langle \fm-\gm,\Thm \fm-\Thm \gm\rangle= \sum_{n=0}^{\infty}(\anm-\bnm)^2(1-\lambdam).
	\end{equation}
	Since $0<\lambdam<1$ \cite{slepian_I, slepian_IV, bertero1998introduction}, we then have that $\|\Thm \fm-\Thm \gm\|^2\leq\langle \fm-\gm,\Thm \fm-\Thm \gm\rangle$. Therefore, $\Thm$ is firmly nonexpansive \cite{bauschke2012firmly}. 
	
	Interestingly, we can prove that $f_k^{(m)}\to h$. Consider the following
	\begin{gather}
	\label{eq:fknorm}
	||f_k^{(m)}||^2 = \sum_{n=0}^{\infty}(\cnm)^2(1-(1-\lambda_n^{(m)})^k)^2
	\end{gather}
	where we have used the series expansion of $h\in\bspace$ with coefficients $\cnm$ and $f_k^{(m)} = h-\sum_{n=0}^{\infty}\cnm(1-\lambdam)^k\pro$ \cite[Theorem 1]{papoulis1975new}. The eigenvalues $\lambdam$ are positive and bounded by 1, thus $(1-(1-\lambda_n^{(m)})^k)^2<1$. Then, the terms in the summation in (\ref{eq:fknorm}) satisfy
	\begin{equation}
	\lim_{k\to\infty}(\cnm)^2(1-(1-\lambda_n^{(m)})^k)^2=(\cnm)^2
	\end{equation}
	and 
	\begin{equation}
	(\cnm)^2(1-(1-\lambda_n^{(m)})^k)^2\leq (\cnm)^2
	\end{equation}
	for $m=1,\ldots,M$ and $n\geq0$. Then, by Tannery's Limiting Theorem \cite[Chapter 3]{loya2017amazing}, it can be concluded that $\lim_{k\to\infty}||f_k^{(m)}||=||h||$. As $\Thm$ is firmly nonexpansive, we have that $f_k^{(m)}\rightharpoonup h$ and by the Radon-Riesz Theorem \cite[Chapter 8]{royden1968real} it follows that $f_k^{(m)}\to h$. 
	
	The operators $T_h^{(m)}$ are $1/2$-averaged \cite{bauschke2011convex}. It readily follows from \cite[Proposition 2.2]{Combettes:2015aa} that $T_h$ is also $1/2$-averaged, i.e., firmly nonexpansive, and $\mathrm{Fix}\ T_h=\bigcap_{m=1}^{M}\mathrm{Fix}\ T_h^{(m)}=h$ \cite[Corollary 5.19]{bauschke2011convex}. From \cite[Corollary 5.17(ii)]{bauschke2011convex}, we have that $T_hf_k-f_k\to0$ which, together with \cite[Theorem 1.2]{Petryshyn:1973aa}, implies that $f_k\to h$.
\end{proof}

% Projected Landweber Iteration
\subsection{Relation to the \acl{pli}}
If we consider recovery from unweighted continuous measurements, i.e., setting each $\omega_m=1$, the iterative reconstruction proposed in (\ref{eq:Iteration}) can be viewed as a Projected Landweber Iteration. This also guarantees nonexpansiveness and strong convergence \cite{eicke1992iteration,sanz1983papoulis}. This can be seen by noticing that the operator $\sum_{m=1}^{M}\Dm$ is self-adjoint and idempotent---thus an orthogonal projection---and that $\PS$ is a projection onto the closed and convex set $\mathcal{B}_{\mathcal{S}}$. It is important to emphasize that this does not apply to the case described above, i.e., we have instead a convex combination of truncation operators $\sum_{m=1}^{M}\omega_m\Dm$ for $\sum_{m=1}^{M}\omega_m=1$ and positive weights.

\subsection{Regularization}
\label{sec:Regularization}
The convergence guarantees of Theorem \ref{theorem:FirmlyStrongly} are relevant in an error-free scenario. In practice, however, there may be many sources of error that can cause the iteration in \eqref{eq:Iteration} to become unstable. For example, a nonideal lowpass filtering operation, the signal $h$ being only approximately bandlimited, or the values of $h$ in the corresponding measurement sets corrupted by numerical errors and noise. We show an example of this instability in Section \ref{sec:simulations}.

In order to make the reconstruction robust, we introduce a regularized version of our method by replacing the operator $T_h$ in \eqref{eq:Iteration} by the following regularized operator
\begin{equation}
	\label{eq:RegOperator}
	\Thh := \sum_{m=1}^{M}\omega_m\Thmh
\end{equation}
where $\Thmh(f):=\PS[(1-\mu\tau)f+\tau\Dm(h-f)]$, $\sum_{m=1}^M\omega_m=1$ for $\omega_m\in(0,1]$, and $\mu,\tau$ are regularization parameters. In the following, we will refer to this as the regularized case. If we choose $\mu>0$ and $0<\tau<2/(1+2\mu)$, it can be shown that the individual operators $\Thmh$ are Banach contractions with Lipschitz constant $1-\mu\tau$ \cite{eicke1992iteration,bertero1998introduction}. Moreover, the unique fixed point of $\Thmh$ is precisely the minimizer of the following Tikhonov functional, i.e.,
\begin{equation}
	\label{eq:TikhonovFunctional}
	\mathrm{Fix}\ \Thmh = \underset{f\in\bspace}{\arg\min}\|\Dm f-h\|^2+\mu\cdot\|f\|^2
\end{equation}
and $f_{k+1}=\Thmh(f_k)$ for a single region is known as Damped Projected Landweber Iteration \cite{eicke1992iteration}.

It turns out that it is possible to draw similar conclusions about the combined operator $\Thh$, i.e., it follows from \cite[Lemma 4.11]{bauschke2012firmly} that $\Thh$ is also a Banach contraction with Lipschitz constant $1-\mu\tau$. Additionally, we show in the next result that the unique fixed point of $\Thh$ is also the unique minimizer of a Tikhonov functional.

\begin{theorem}
	If $\hat{f}^*\in\bspace$ is the unique fixed point of $\Thh$, i.e., $\Thh\hat{f}^*=\hat{f}^*$, then $\hat{f}^*$ is also the unique minimizer of
	\begin{equation}
		\label{eq:OptConvexCombination}
		\hat{f}^*=\underset{f\in\bspace}{\arg\min}\sum_{m=1}^M\omega_m\|\Dm f-h\|^2+\mu\cdot\|f\|^2
	\end{equation}
	for $\mu>0$.
\end{theorem}
\begin{proof}
	The functional defined by $\gamma(f):=\sum_{m=1}^M\omega_m\|\Dm f-h\|^2+\mu\cdot\|f\|^2$ for $f\in \lspace(\mathbb{R}^N)$ is strictly convex for $\mu>0$. Then, it is necessary and sufficient for $\hat{f}$ to be the minimizer of the unconstrained problem that the functional derivative $\gamma(\hat{f}^*)'f=0$ for all $f\in \lspace(\mathbb{R}^N)$. Following a reasoning similar to \cite[Theorem 5.2]{engl1996regularization}, we have that
	\begin{equation}
	\begin{split}
	\gamma(\hat{f})'f=2\sum_{m=1}^M\omega_m\langle \Dm(h-\hat{f})+\mu\hat{f},f\rangle=0.
	\end{split}
	\end{equation}
	for all $f\in \lspace(\mathbb{R}^N)$. Since we are looking for solutions constrained to be bandlimited, the minimizer is found by means of a projection onto the closed convex set $\bspace$, i.e.,
	\begin{equation}
	\mu\hat{f}^*+\PS\sum_{m=1}^{M}\omega_m\Dm(h-\hat{f}^*)=0
	\end{equation}
	which is equivalent to the fixed point property $\Thh\hat{f}^*=\hat{f}^*$.
\end{proof}

\subsection{Truncated \acs{pswf} expansion}
If the bandlimited functions of interest further satisfy that they can be represented by a finite number of prolate coefficients, there are several results that follow directly from the previous sections. Let us first formally introduce this set of functions as follows
\begin{equation}
	\bspace^N = \{f\in\bspace:f=\sum_{n=0}^N\anm\pro, N>0,\anm\in\ell^2(\mathbb{R})\}.
\end{equation}

\subsubsection{Unregularized case}
Even without regularization, we have a contraction mapping with the corresponding stability and convergence guarantees \cite[Theorem 1.50]{bauschke2011convex}. The next result shows that $\Thm:\bspace^N\to\bspace^N$ is a Banach contraction. In consequence \cite[Lemma 4.11]{bauschke2012firmly}, $T_h:\bspace^N\to\bspace^N$ is also a Banach contraction with Lipschitz constant $\sum_{m=1}^{M}\omega_m(1-\lambda_N^{(m)})$.  
\begin{proposition}
	\label{prop:FiniteContractionUnreg}
	The operator $\Thm:\bspace^N\to\bspace^N$ is a Banach contraction with Lipschitz constant $1-\lambda_N^{(m)}$.
\end{proposition}
\begin{proof}
	Similar to (\ref{eq:firmly_left}), we have that
	\begin{equation}
	\begin{split}
	||\Thm f-\Thm g||^2&=\sum_{n=0}^N|\anm-\bnm|^2(1-\lambdam)^2\\
	&\leq(1-\lambda_N^{(m)})^2||f-g||^2
	\end{split}
	\end{equation}
	where the inequality follows from the fact that the corresponding eigenvalues form a decreasing sequence \cite{slepian_I,slepian_IV, bertero1998introduction}. Since $0<\lambdam<1$ for all $n,m$ \cite{slepian_I,slepian_IV,bertero1998introduction}, then $0<1-\lambda_N^{(m)}<1$ which guarantees a contraction mapping.
\end{proof}

\subsubsection{Regularized case}
Even though the unregularized case is already a contraction, we can include some regularization to decrease the Lipschitz constant. This can result in a faster convergence rate at the expense of accuracy in the reconstruction, i.e., convergence to the original function is not guaranteed.
\begin{proposition}
	\label{prop:FiniteContraction}
	The operator $\Thmh:\bspace^N\to\bspace^N$ is a Banach contraction with Lipschitz constant $|1-\tau(\lambda_N^{(m)}+\mu)|$.
\end{proposition}
\begin{proof}
	Following the same reasoning as in Proposition \ref{prop:FiniteContractionUnreg}, we can readily write
	\begin{equation}
		\label{eq:cderivalt}
		\|\Thmh f-\Thmh g\|^2 \leq|1-\tau(\lambda_N^{(m)}+\mu)|^2||f-g||^2
	\end{equation}
	where the inequality follows from
	\begin{equation}
		\label{eq:TruncatedRegularized}
		|1-\tau(\lambda_n^{(m)}+\mu)|^2\leq|1-\tau(\lambda_N^{(m)}+\mu)|^2<1
	\end{equation}
	whenever $0<\tau\leq2/(\lambda_0^{(m)}+\lambda_N^{(m)}+2\mu)$ and the fact that $1>\lambda_0^{(m)}>\ldots>\lambda_N^{(m)}$. Using the right-hand side of (\ref{eq:TruncatedRegularized}), we conclude that it is a contraction mapping.
\end{proof}

An interesting observation is that the Lipschitz constant---for a fixed $\mu$---is small whenever the difference $\lambda_0^{(m)}-\lambda_N^{(m)}$ is small. This can be seen by using the upper bound of $\tau$ in the Lipschitz constant of Proposition \ref{prop:FiniteContraction}, i.e.,
\begin{equation}
	1-2\frac{\lambda_N^{(m)}+\mu}{\lambda_0^{(m)}+\lambda_N^{(m)}+2\mu}.
\end{equation}
As before, it is straightforward to conclude that the combined operator $\Thh:\bspace^N\to\bspace^N$ is a Banach contraction with Lipschitz constant $\sum_{m=1}^{M}\omega_m|1-\tau(\lambda_N^{(m)}+\mu)|$.

\vspace{5pt}

\subsubsection{Weighting of signal segments}
It is interesting to mention that there is a connection between $\lambda_N^{(m)}$ and the size of $\mathcal{D}_m$ that may impact the convergence rate. This may help to choose a different weighting depending on the size of the regions in order to possibly improve convergence. In particular, the larger the size of $\mathcal{D}_m$, the smaller the Lipschitz constant due to increasing $\lambda_N^{(m)}$ \cite{slepian_I, slepian_IV}. As a result, it may be advantageous to assign larger weights to large $\mathcal{D}_m$. Similar conclusions can be drawn when $N$ decreases, which means that convergence can be faster on a particular region the more concentrated the function is on that region.

\section{Numerical Examples}
\label{sec:simulations}
We illustrate the convergence properties of the proposed iterative method in the regularized case, i.e., using $\Thh$ instead of $T_h$ in \eqref{eq:Iteration} and in the unregularized case, using the iteration in \eqref{eq:Iteration} directly. The performance is evaluated by means of the \ac{nmse}, i.e.,
\begin{equation}
	\mathrm{NMSE} = 10\cdot\log_{10}\left(\frac{\|h-f_e\|^2}{\|h\|^2}\right)
\end{equation}
where $h$ is the original and $f_e$ is the extrapolated signal. As an example, we construct a two-dimensional bandlimited signal $h$ such that $\hat{P}_\mathcal{S}h=h$ where $\hat{P}_\mathcal{S}$ represents a truncation operation in the \ac{dft} domain.

\subsection{Performance}
The signal $h$ is sampled in the regions $\mathcal{D}_1$ to $\mathcal{D}_6$ (see Fig. \ref{fig:Combined}). We apply the iteration in \eqref{eq:Iteration}, using a uniform weighting as well as the operator $\hat{P}_\mathcal{S}$. Fig. \ref{fig:Combined} shows the original signal $h$ and the reconstructed signal $f_e$ after 1000 iteration steps, resulting in an \ac{nmse} of approximately \SI{-21.6}{\decibel}. The right-hand side shows the \ac{nmse} over the number of iterations. 

Naturally, the performance can be improved by using multiple continuous measurements, as more information about the original signal is used. However, it is difficult to make general statements regarding the performance change depending on the total support that is covered by the regions, as the performance heavily depends on the signal in the regions.

\begin{figure}
	\centering
	\includegraphics[width = 0.49\textwidth]{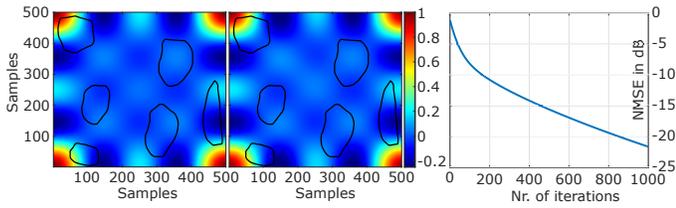}
	\caption{Left: Original signal $h$ together with the regions $\mathcal{D}_1,\ldots,\mathcal{D}_6$ used for reconstruction. Middle: Reconstructed signal $f_e$ after 1000 steps of iteration \eqref{eq:Iteration} with the corresponding regions at an \ac{nmse} of approximately \SI{-21.6}{\decibel}. Right: \ac{nmse} over the number of iterations}
	\label{fig:Combined}
\end{figure}

\subsection{Stability}
The need for a regularized version can be illustrated by using an example that incorporates nonidealities. In particular, we assume that the signal of interest is not perfectly bandlimited. We construct it by adding Gaussian functions to the signal $h$ such that it is nonbandlimited in the \ac{dft} domain, i.e., $\hat{P}_\mathcal{S}h\neq h$. The \ac{snr}, measuring the ratio between signal and out-of-band noise energy, has been set to approximately \SI{6.9}{\decibel} which is an arbitrary choice for illustration. Fig. \ref{fig:Regularization} shows how the \ac{nmse} in the unregularized case seems to grow without bound. In contrast, the regularized version of iteration \eqref{eq:Iteration} leads to an approximately constant \ac{nmse} of \SI{-8.9}{\decibel} after 100 iterations and therefore promotes stable convergence of the iteration.

\begin{figure}
	\centering
	\input{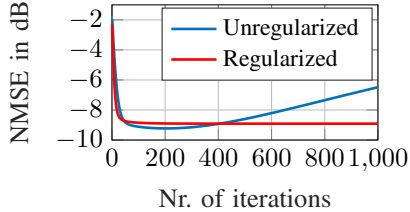}
	\caption{Reconstruction error for an approximately bandlimited signal at a \ac{snr} of approximately \SI{6.9}{\decibel} where the regularized case corresponds to iteration \eqref{eq:Iteration} using $\Thh$ and the unregularized case directly to iteration \eqref{eq:Iteration}. Four non-overlapping regions with a weight of $\omega_m=1/4$ for all $m$ and regularization parameters of $\mu=0.005$ and $\tau = 1.99/(1+2\mu)\approx1.97$ have been used.}
	\label{fig:Regularization}
\end{figure}

\section{Conclusion}
We proposed an iterative method to extrapolate multidimensional bandlimited signals from individually weighted continuous measurements on several bounded regions, i.e., truncated versions of the original signal. This can be seen as an extension of \cite{papoulis1975new}. We introduced regularization to the algorithm in order to stabilize it in the presence of errors or nonidealities. We proved convergence of the unregularized iteration to the original signal and showed that the regularized method converges to the solution of an optimization problem that is related to Tikhonov regularization. We discussed stability and convergence properties of our method in the case that the underlying functions can be represented by a truncated PSWF expansion and established a connection between the weights of the signal segments and the eigenvalues of PSWFs. We also illustrated the method in simulations and showed that the instability in the presence of errors can be resolved by using the regularized version of our method.

%\section*{Acknowledgment}

%\section*{References}

\bibliographystyle{IEEEtran}
\bibliography{strings}

\begin{thebibliography}{10}
\providecommand{\url}[1]{#1}
\csname url@samestyle\endcsname
\providecommand{\newblock}{\relax}
\providecommand{\bibinfo}[2]{#2}
\providecommand{\BIBentrySTDinterwordspacing}{\spaceskip=0pt\relax}
\providecommand{\BIBentryALTinterwordstretchfactor}{4}
\providecommand{\BIBentryALTinterwordspacing}{\spaceskip=\fontdimen2\font plus
\BIBentryALTinterwordstretchfactor\fontdimen3\font minus
  \fontdimen4\font\relax}
\providecommand{\BIBforeignlanguage}[2]{{%
\expandafter\ifx\csname l@#1\endcsname\relax
\typeout{** WARNING: IEEEtran.bst: No hyphenation pattern has been}%
\typeout{** loaded for the language `#1'. Using the pattern for}%
\typeout{** the default language instead.}%
\else
\language=\csname l@#1\endcsname
\fi
#2}}
\providecommand{\BIBdecl}{\relax}
\BIBdecl

\bibitem{Shannon:1949aa}
C.~E. Shannon, ``Communication in the presence of noise,'' \emph{Proceedings of
  the IRE}, vol.~37, no.~1, pp. 10--21, 1949.

\bibitem{Yen:1956aa}
J.~Yen, ``On nonuniform sampling of bandwidth-limited signals,'' \emph{IRE
  Transactions on Circuit Theory}, vol.~3, no.~4, pp. 251--257, 1956.

\bibitem{Paley:1934aa}
R.~E. A.~C. Paley and N.~Wiener, \emph{Fourier transforms in the complex
  domain}.\hskip 1em plus 0.5em minus 0.4em\relax American Mathematical Soc.,
  1934, vol.~19.

\bibitem{Eldar:2015aa}
Y.~C. Eldar, \emph{Sampling Theory: Beyond Bandlimited Systems}.\hskip 1em plus
  0.5em minus 0.4em\relax Cambridge University Press, 2015.

\bibitem{Williams:1999aa}
E.~G. Williams, \emph{Fourier acoustics: sound radiation and nearfield
  acoustical holography}.\hskip 1em plus 0.5em minus 0.4em\relax Academic
  press, 1999.

\bibitem{Born:2013aa}
M.~Born and E.~Wolf, \emph{Principles of optics: electromagnetic theory of
  propagation, interference and diffraction of light}.\hskip 1em plus 0.5em
  minus 0.4em\relax Elsevier, 2013.

\bibitem{schleicher2007seismic}
J.~Schleicher, M.~Tygel, and P.~Hubral, \emph{Seismic true-amplitude
  imaging}.\hskip 1em plus 0.5em minus 0.4em\relax Society of Exploration
  Geophysicists, 2007.

\bibitem{Grande:2019aa}
E.~F. Grande, ``Sound field reconstruction in a room from spatially distributed
  measurements,'' in \emph{23rd International Congress on Acoustics}.\hskip 1em
  plus 0.5em minus 0.4em\relax German Acoustical Society (DEGA), 2019, pp.
  4961--68.

\bibitem{Chai:2000aa}
J.-X. Chai, X.~Tong, S.-C. Chan, and H.-Y. Shum, ``Plenoptic sampling,'' in
  \emph{Proceedings of the 27th annual conference on Computer graphics and
  interactive techniques}, 2000, pp. 307--318.

\bibitem{Stein:2003aa}
E.~M. Stein and R.~Shakarchi, \emph{Complex analysis}, ser. Princeton Lectures
  in Analysis, II.\hskip 1em plus 0.5em minus 0.4em\relax Princeton University
  Press, Princeton, NJ, 2003.

\bibitem{papoulis1975new}
A.~Papoulis, ``A new algorithm in spectral analysis and band-limited
  extrapolation,'' \emph{IEEE Transactions on Circuits and Systems}, vol.~22,
  no.~9, pp. 735--742, 1975.

\bibitem{cadzow1979extrapolation}
J.~Cadzow, ``An extrapolation procedure for band-limited signals,'' \emph{IEEE
  Transactions on Acoustics, Speech, and Signal Processing}, vol.~27, no.~1,
  pp. 4--12, 1979.

\bibitem{tom1981convergence}
V.~Tom, T.~Quatieri, M.~Hayes, and J.~McClellan, ``Convergence of iterative
  nonexpansive signal reconstruction algorithms,'' \emph{IEEE Transactions on
  Acoustics, Speech, and Signal Processing}, vol.~29, no.~5, pp. 1052--1058,
  1981.

\bibitem{donoho1989uncertainty}
D.~L. Donoho and P.~B. Stark, ``Uncertainty principles and signal recovery,''
  \emph{SIAM Journal on Applied Mathematics}, vol.~49, no.~3, pp. 906--931,
  1989.

\bibitem{landau1961recovery}
H.~J. Landau and W.~L. Miranker, ``The recovery of distorted band-limited
  signals,'' \emph{Journal of Mathematical Analysis and Applications}, vol.~2,
  no.~1, pp. 97--104, 1961.

\bibitem{youla1982image}
D.~C. Youla and H.~Webb, ``Image restoration by the method of convex
  projections: Part 1 theory,'' \emph{IEEE transactions on medical imaging},
  vol.~1, no.~2, pp. 81--94, 1982.

\bibitem{sezan1982image}
M.~I. Sezan and H.~Stark, ``Image restoration by the method of convex
  projections: Part 2-applications and numerical results,'' \emph{IEEE
  Transactions on Medical Imaging}, vol.~1, no.~2, pp. 95--101, 1982.

\bibitem{schafer1981constrained}
R.~W. Schafer, R.~M. Mersereau, and M.~A. Richards, ``Constrained iterative
  restoration algorithms,'' \emph{Proceedings of the IEEE}, vol.~69, no.~4, pp.
  432--450, 1981.

\bibitem{bauschke2012firmly}
H.~H. Bauschke, S.~M. Moffat, and X.~Wang, ``Firmly nonexpansive mappings and
  maximally monotone operators: correspondence and duality,'' \emph{Set-Valued
  and Variational Analysis}, vol.~20, no.~1, pp. 131--153, 2012.

\bibitem{slepian_I}
D.~Slepian and H.~O. Pollak, ``Prolate spheroidal wave functions, fourier
  analysis and uncertainty—i,'' \emph{Bell System Technical Journal},
  vol.~40, no.~1, pp. 43--63, 1961.

\bibitem{slepian_IV}
D.~Slepian, ``Prolate spheroidal wave functions, fourier analysis and
  uncertainty—iv: extensions to many dimensions; generalized prolate
  spheroidal functions,'' \emph{Bell System Technical Journal}, vol.~43, no.~6,
  pp. 3009--3057, 1964.

\bibitem{bertero1998introduction}
M.~Bertero and P.~Boccacci, \emph{Introduction to inverse problems in
  imaging}.\hskip 1em plus 0.5em minus 0.4em\relax CRC press, 1998.

\bibitem{loya2017amazing}
P.~Loya, \emph{Amazing and aesthetic aspects of analysis}.\hskip 1em plus 0.5em
  minus 0.4em\relax Springer, 2017.

\bibitem{royden1968real}
H.~L. Royden, \emph{Real analysis}.\hskip 1em plus 0.5em minus 0.4em\relax
  Krishna Prakashan Media, 1968.

\bibitem{bauschke2011convex}
H.~H. Bauschke, P.~L. Combettes \emph{et~al.}, \emph{Convex analysis and
  monotone operator theory in Hilbert spaces}, ser. CMS Books in
  Mathematics.\hskip 1em plus 0.5em minus 0.4em\relax Springer, 2011, vol. 408.

\bibitem{Combettes:2015aa}
P.~L. Combettes and I.~Yamada, ``Compositions and convex combinations of
  averaged nonexpansive operators,'' \emph{Journal of Mathematical Analysis and
  Applications}, vol. 425, no.~1, pp. 55--70, 2015.

\bibitem{Petryshyn:1973aa}
W.~V. Petryshyn and T.~E. Williamson~Jr., ``Strong and weak convergence of the
  sequence of successive approximations for quasi-nonexpansive mappings,''
  \emph{Journal of Mathematical Analysis and Applications}, vol.~43, no.~2, pp.
  459--497, 1973.

\bibitem{eicke1992iteration}
B.~Eicke, ``Iteration methods for convexly constrained ill-posed problems in
  hilbert space,'' \emph{Numerical Functional Analysis and Optimization},
  vol.~13, no. 5-6, pp. 413--429, 1992.

\bibitem{sanz1983papoulis}
J.~Sanz and T.~Huang, ``On the papoulis-gerchberg algorithm,'' \emph{IEEE
  Trans. Circuits and Systems}, p. 907, 1983.

\bibitem{engl1996regularization}
H.~W. Engl, M.~Hanke, and A.~Neubauer, \emph{Regularization of inverse
  problems}.\hskip 1em plus 0.5em minus 0.4em\relax Springer Science \&
  Business Media, 1996, vol. 375.

\end{thebibliography}

\end{document}